\documentclass[11pt,a4paper]{article}

\usepackage{amsmath,amsthm,amsfonts,amssymb,epsfig}
\usepackage{setspace}
\usepackage{xcolor}

\usepackage{amsmath, amsthm, mathscinet}
\usepackage{amssymb}
\usepackage{
	comment,
	enumitem, 
	float, 
	wrapfig,
	subcaption }
\usepackage{microtype}
\usepackage[applemac]{inputenc}
\usepackage{graphicx}
\usepackage{color}
\usepackage{tikz} \usetikzlibrary{trees}
\usepackage{enumitem}
\setlist[enumerate,1]{label=(\roman*)}

\usetikzlibrary{matrix}

\makeatletter
\newcommand*\bigcdot{\mathpalette\bigcdot@{.6}}
\newcommand*\bigcdot@[2]{\mathbin{\vcenter{\hbox{\scalebox{#2}{$\m@th#1\bullet$}}}}}
\makeatother

\def\E{{\mathbb E}}

\def\I{{\cal I}}

\def\P{{\mathbb P}}

\def\R{{\mathbb R}}

\def\I\nd{{\mathbb I}}

\renewcommand{\subset}{\subseteq}

\newtheorem{theorem}{Theorem}[section]

\newtheorem{corollary}[theorem]{Corollary}

\newtheorem{lemma}[theorem]{Lemma}

\newtheorem{proposition}[theorem]{Proposition}

\newtheorem{definition}[theorem]{Definition}

\newtheorem{assumption}[theorem]{Assumption}
\newtheorem{remark}[theorem]{Remark}
\newtheorem{example}[theorem]{Example}

\renewcommand{\epsilon}{\varepsilon}
\newcommand{\nd}{\text{nd}}

\numberwithin{equation}{section}

\title{From Bachelier to Dupire via Optimal Transport
}

 \author{Mathias Beiglb\"ock, Gudmund Pammer,
Walter Schachermayer}
\begin{document}

\maketitle
\begin{abstract} 
Famously mathematical finance was started by Bachelier  in his 1900 PhD thesis where -- among many other achievements -- he also provides a formal derivation of the Kolmogorov forward equation. This forms also the basis for Dupire's (again formal) solution to the problem of finding an arbitrage free model calibrated to the volatility surface. 
 The later result has  rigorous counterparts in the theorems of Kellerer and Lowther. In this survey article we revisit these hallmarks of stochastic finance,  highlighting the role played by some optimal transport results in this context.


\end{abstract}



\section{Bachelier's work relating Brownian motion to mass transport and the heat equation} \label{ahnbwrbmtthe}


In this section, which is mainly dedicated to the historic point of view,  {we follow \cite{Sc03}} and point out that Bachelier already had some thoughts on \textit{``horizontal transport of probability measures''} in his dissertation \textit{``Th{\'e}orie de la sp{\'e}culation''} \cite{Ba00}, which he defended in Paris in 1900.

In this thesis he was the first to consider a \textit{mathematical model} of Brownian motion. Bachelier argued using infinitesimals by visualizing Brownian motion $(W(t))_{t \geqslant 0}$ as an infinitesimal version of a random walk. His 19th century style argument runs as follows. Suppose that the grid in space is given by 
\begin{equation} \label{grid}
\ldots, \ x_{n-2}, \ x_{n-1}, \ x_{n}, \ x_{n+1}, \ x_{n+2}, \ \ldots
\end{equation}
having the same (infinitesimal) distance $\Delta x = x_{n}-x_{n-1}$, for all $n$, and such that at time $t$ these points have (infinitesimal) probabilities 
\begin{equation} \label{gridp}
\ldots, \ p_{n-2}^{t}, \ p_{n-1}^{t}, \ p_{n}^{t}, \ p_{n+1}^{t}, \ p_{n+2}^{t}, \ \ldots
\end{equation}
under the random walk under consideration. What are the probabilities 
\begin{equation} \label{gridpp}
\ldots, \ p_{n-2}^{t+\Delta t}, \ p_{n-1}^{t+\Delta t}, \ p_{n}^{t+\Delta t}, \ p_{n+1}^{t+\Delta t}, \ p_{n+2}^{t+\Delta t}, \ \ldots
\end{equation}
of these points at time $t+\Delta t$? 

The random walk moves half of the mass $p_{n}^{t}$, sitting on $x_{n}$ at time $t$, to the point $x_{n+1}$. En revanche, it moves half of the mass $p_{n+1}^{t}$, sitting on $x_{n+1}$ at time $t$, to the point $x_{n}$. We thus may calculate the net difference between $p_{n}^{t}/2$ and $p_{n+1}^{t}/2$, which Bachelier identifies with 
\begin{equation}
-\tfrac{1}{2} \, \frac{\partial p^{t}}{\partial x}(x), 
\end{equation}
where we let $x = x_n = x_{n+1}$ which is legitimate for Bachelier as $x_n$ and $x_{n+1}$ only differ by an infinitesimal.

This amount of mass is transported from the interval $(-\infty,x_{n}]$ to $[x_{n+1},\infty)$ during the time interval $ (t,t+\Delta t)$. In Bachelier's own words, this is very nicely captured by the following quote from his thesis: 

\medskip

\textit{``Each price $x$ during an element of time radiates towards its neighboring price an amount of probability proportional to the difference of their probabilities. I say proportional because it is necessary to account for the relation of $\Delta x$ to $\Delta t$. The above law can, by analogy with certain physical theories, be called the law of radiation or diffusion of probability.''}

\medskip

Passing formally to the continuous limit and -- using today's terminology -- denoting by 
\begin{equation} \label{dist}
P_t(x) = \int_{-\infty}^{x} p_t(z) \, \textnormal{d}z
\end{equation}
the distribution function associated to the Gaussian density function $p_t(x)$, Bachelier thus deduces in this intuitively convincing way the relation
\begin{equation} \label{lbhef}
\frac{\partial P}{\partial t} = \frac{1}{2} \frac{\partial p}{\partial x},
\end{equation}
where we have normalized the relation between $\Delta x$ and $\Delta t$ to obtain the constant $1/2$. By differentiating 
(\ref{lbhef}) with respect to $x$ one obtains the usual heat equation 
\begin{equation} \label{heflbhef}
\frac{\partial p}{\partial t} = \frac{1}{2} \frac{\partial^{2} p}{\partial x^{2}}
\end{equation}
for the density function $p(t,x)$. 
Of course, the heat equation was known to Bachelier, and he notes regarding (\ref{heflbhef})
 \textit{``C'est une {\'e}quation de Fourier.''}

Bachelier has thus derived, on a formal level, the Kolmogorov forward equation, also known as Fokker-Planck-equation, for the propagation of a probability density $p$ under Brownian motion. The forward equation will also play an important role subsequently and we take the opportunity to note that Bachelier's argument can equally well be applied to the more general process with increments $dX_t = \sigma(t,X_t)\, dW_t$ to arrive at the PDE
\begin{equation}\label{Fokker0}
\frac{\partial}{\partial t} P= \frac{1}{2} \frac{\partial}{\partial x} \left(\sigma^2 p\right),
\qquad \qquad
\frac{\partial}{\partial t}p = \frac12 \frac{\partial^2}{\partial x^2}\left(\sigma^2 p\right).
\end{equation}

But let us still remain with the form (\ref{lbhef})
 of the heat equation and analyze its message in terms of \textit{``horizontal transport of probability measures''}. One may ask: what is the ``velocity field'', acting on the set of probabilities on $\R$, which moves the probability density $p_t(\cdot)$ to the probability density $p_{t+dt}(\cdot)$? Following Bachelier's intuition and keeping in mind that the mass sitting at time $t$ on $x$ equals $p_t(x)$, the velocity of this move at the point $x$ has to be equal to
 
 \begin{equation} \label{score}
 -\frac{1}{2} \frac{\frac{\partial p_{t}}{\partial x}(x)}{p_t(x)} 
\end{equation}
which has the natural interpretation as the ``speed'' of the horizontal  transport induced by $p_t(x)$. We thus encounter \textit{in nuce} the  ``score function'' $\frac{p'_t(x)}{p_t(x)} = \frac{\nabla p_t(x)}{p_t(x)}$ where the nabla notation $\nabla$ indicates that this is a vector field which makes perfect sense in the $n$-dimensional case too.

At this stage we can relate Bachelier's work with the more recent notion of the \textit{Wasserstein metric} $\mathcal{W}_2 (\cdot,\cdot)$, at least intuitively and at an infinitesimal level. One may ask: what is the necessary kinetic energy needed to transport $p_t(\cdot)$ to $p_{t+dt}(\cdot)$? Knowing the speed (\ref{score}) and the usual formula for the kinetic energy, we obtain for the Wasserstein distance between the two infinitesimally close probabilities $p_t$ and $p_{t+dt}$

\begin{equation}\label{Fisher}
 \frac{\mathcal{W}_2 \left(p_t,p_{t+dt}\right)}{dt} := \frac{1}{2} \left(\int_{\R} \left(\frac{p'_t(x)}{p_t(x)}\right)^2 dx \right)^{\frac{1}{2}}
\end{equation}

For a formal definition of the Wasserstein distance $\mathcal{W}_2 (\cdot,\cdot)$ we refer, e.g., to \cite{Vi09}. While for the finite version of the Wasserstein distance between two probability measures one has to find an optimal transport plan, the situation is simpler -- and very pleasant -- in the case of the infinitesimal transport induced by the vector field (\ref{score}). This infinitesimal transport is automatically optimal.
Intuitively this corresponds to the geometric insight in the one-dimensional case that the transport lines of infinitesimal length cannot cross each other. For a thorough treatment of the geometry of absolutely continuous curves of probabilities such as $\left(p_t(\cdot)\right)_{t \ge 0}$ above we refer to the lecture notes \cite{AmGiSa08}.

\medskip


\medskip

We finish the section by returning to Bachelier's thesis. The \textit{rapporteur} of Bachelier's dissertation was no lesser a figure than Henri Poincar{\'e}. Apparently he was aware of the enormous potential of the section \textit{``Rayonnement de la probabilit{\'e}''} in Bachelier's thesis, when he added to his very positive report the handwritten phrase: \textit{``On peut regretter que M. Bachelier n'ait pas d{\'e}velopp{\'e} davantage cette partie de sa th{\`e}se.''} That is: One might regret that Mr. Bachelier did not develop further this part of his thesis. Truly prophetic words!

\section{Dupire's Formula} \label{Dupire}

We now turn to a well-known and more recent topic in Mathematical Finance continuing the early achievements of Bachelier.

Suppose that in a financial market we know the prices of ``many'' European options on a given (highly liquid) stock $S$. What can we deduce from this data about the prices of exotic, i.e., path-dependent options?

This question leads to the following mathematical idealization: suppose we know the prices of \textit{all} European call options, i.e., the price $C(t,x)$ of every call option with strike price $x$ and maturity $t$, for every $0 \le t \le T$ and $x \in \R_+$. Our task is to analyze the set of all possible (local) martingale measures for the stock price processes which are compatible with this data.
Once we have a hand on the relevant set of martingale measures, we can price arbitrary exotic options by taking expectations.

To make the question meaningful, it is a good idea to restrict the class of processes under consideration, e.g., to continuous, Markovian martingales.

We make the economically meaningful assumptions that the function $(t,x) \mapsto C(t,x)$ is sufficiently smooth, as well as strictly convex in the variable $x$ and strictly increasing in the variable $t$, to allow for the subsequent formal manipulations. 


The first observation is that the knowledge of $C(t,x)$, for $0\le t\le T$ and $x>0$ is tantamount to the knowledge of the marginal probabilities $(\mu_t)_{0 \le t \le T}$ of the underlying stock price process under a martingale measure which determines the prices, via the formula 
\begin{equation} \label{eq3.1}
C(t,x) = \E_{\mu_t}[(S_t-x)_+].
\end{equation}
This observation goes back to the Breeden and Litzenberger \cite{BrLi78}.

If the measures $\mu_t$ are absolutely continuous with respect to Lebesgue measure with a continuous density function $p_t(x)$ then (\ref {eq3.1}) amounts to the relation
\begin{equation}\label{eq3.2}
p_t(x) = C_{xx}(t,x),\qquad x>0,
\end{equation}
as one verifies via integration by parts.

In a very influential and highly cited paper from 1994 \cite{Du94} (compare also the work of Derman and Kani \cite{DeKa94}) Dupire  considered diffusion processes of the form 
\begin{equation}\label{eq3.3}
\frac{dS_t}{S_t}=\sigma(t,S_t)\, dW_t,\qquad 0\le t \le T,
\end{equation}
where the ``local volatility'' $\sigma(\cdot,\cdot)$ is modeled as a {\it deterministic} function  of $t$ and $x$, and $W_t$ is a Brownian motion adapted to its natural filtration $(\mathcal F_t)_{0 \le t \le T }$. It turns out that there is the beautiful and strikingly simple ``Dupire's formula'' which relates $\sigma(\cdot,\cdot)$ to the given option prices $C(t,x)$, namely
\begin{equation}\label{eq3.4}
 \frac{\sigma^2(t,x)}{2} = \frac{C_t(t,x)}{x^2 C_{xx}(t,x)} .
\end{equation}
Indeed, the Fokker-Planck equation implies -- at least on a formal level (cf.\ (\ref{Fokker0})) -- that $p_t(x)$ satisfies the PDE

\begin{equation}\label{Fokker}
\frac{\partial}{\partial t}p_t(x) = \frac{\partial^2}{\partial x^2}\left(\frac{\sigma^2(t,x)}{2}p_t(x)\right)
\end{equation}
Integrating in $x$, using (\ref{eq3.2}), and changing the order of derivatives quickly yields (\ref{eq3.4}).

We note that this beautiful argument is very much in line with Bachelier's reasoning in (\ref{lbhef}) and (\ref{heflbhef}) above pertaining to the case of constant volatility $\sigma$. We note in passing that Bachelier used instead of the wording ``volatility'' the more colorful term ``nervousness of the market''. 

\vspace{1em}



Of course, Dupire's formal arguments need proper regularity assumptions in order to be justified. There are two aspects: existence and uniqueness of the martingales fitting the given option prices $C(t,x)$. As regards the former, the question of existence amounts to a remarkable theorem by Kellerer \cite{Ke72, Ke73}: 
given a family $(\mu_t)_{0 \le t \le T}$ of probability distributions on $\R$ which is  increasing in convex order, there is a Markov martingale having these probabilities as marginals. 
By increasing in convex order we mean that each $\mu_t$ has finite first moment and that $\mu_t (f)$ is nondecreasing in $t$, for every convex function $f$ on $\R$. 
Kellerer's theorem extends earlier work of Strassen \cite{St65} who establish the discrete time version of the result. We also note that the convex order condition on the marginal distributions is necessary as easily follows from Jensen's theorem.


Kellerer's theorem goes far beyond the simple formula (\ref{eq3.4}) and has been further refined, notably by Lowther in an impressive series of papers \cite{Lo08a, Lo08b, Lo08d}. We shall review these results in the subsequent sections.

However, from an application point of view, the existence question is not of primordial relevance. After all, the function $C(t,x)$ is an idealization of reality which has to be estimated from a finite set of given European option prices. In this context, it does not harm to make strong regularity assumptions on the smoothness and convexity (in the variable $x$) of the function $C(t,x)$  which justify the above argument. Under such assumptions Dupire's solution (\ref{eq3.4}) does make sense and the issue of existence is settled.

A different issue is the question of uniqueness. As we shall see below, this question is challenging and relevant -- at least from a mathematical point of view -- even in very regular settings, such as the Bachelier or the Black-Scholes model.

In order to formulate existence and uniqueness results for a process with given marginals, one has to specify the class of processes with respect to which we want to establish existence and uniqueness. Under proper regularity assumptions the unique solution should, of course, equal Dupire's solution. Dupire's process is a \textit{martingale with continuous paths, enjoying the Markov property}. Is Dupire's solution unique within this class? In a veritable tour de force Lowther  \cite{Lo08a, Lo08b} has shown  that the answer is yes, provided that we replace the word \textit{Markov} by the word \textit{strong Markov} and one restricts to continuous processes.

 We also refer to Theorem 6.1 in \cite{HiRoYo14} where a slightly different version of this theorem, credited to Pierre, is proved. These theorems settle the question of uniqueness in a very satisfactory way. We shall discuss Lowther's theorem in more detail in Section 6.

To the best of our knowledge the following question remained open: is it really necessary to add the adjective \textit{strong} to the word \textit{Markov} in Lowther's uniqueness theorem? At least, if one is willing to accept strong regularity assumptions on the function $C(\cdot,\cdot)$ and the resulting process $S$ as defined in (\ref{eq3.3}) one may ask whether the Markov property alone is sufficient. We will focus on this question in the next section. 



\section{An eye-opening example}

The subsequent example is known since the work by Dynkin and Jushkevich in the fifties  (see \cite{DyJu56}). 

\begin{example}\label{easy}

There is an $\R_+$-valued, continuous, Markov martingale which fails to be strongly Markovian.
\end{example}

\begin{proof}\label{easyproof}
We define the process $S=(S_t)_{0\le t\le 1}$ by starting at $S_0 = 1$ and subsequently proceeding in two steps. For $t \in [0,\frac{1}{2}]$ the process $S$ is a stopped geometric Brownian motion, i.e.,
\begin{equation} \label{geom}
S_t = \exp \left(B_{t \wedge \tau} -\frac{t \wedge \tau}{2} \right) \qquad \qquad 0\le t\le \frac{1}{2},
\end{equation}
where $B$ is a standard Brownian motion and $\tau$ is the first moment when $S$ hits the level $2$. 

For $\frac{1}{2} \le t \le1$, we distinguish two cases. If $S$ has been stopped, i.e., if $S_{\frac{1}{2}} = 2$, the process $S$ simply remains constant at the level $2$. If this is not the case, the process continues to follow a geometric Brownian motion, i.e.,
\begin{equation} \label{geom2}
S_t = \exp \left(B_t -\frac{t}{2} \right), \qquad \qquad \frac{1}{2} \le t \le 1 \quad \text{on the set} \quad  \left\{\tau > \frac{1}{2}\right\}  ,
\end{equation}

Obviously $S$ is a continuous martingale. The crucial feature is its Markovian nature: the Markov property follows from the fact that, for every fixed (deterministic) time $\frac{1}{2}\le t\le 1$, the probability for the geometric Brownian motion $S_t$  to be equal to $2$ is zero on the set $\{\tau>\frac{1}{2}\}$. Hence, for every fixed $\frac{1}{2} \le t \le 1$, the conditional law of $(S_u)_{t \le u \le 1}$ is \textit{almost surely} determined by the present value $S_t$ of the process.

Why does $S$ fail to be strongly Markovian? On the set $\{\tau > \frac{1}{2}\}$ define the stopping time $\vartheta$ as the first instance $u > \frac{1}{2}$ when $S_u$ equals the value $2$, which happens with positive probability during the interval $\left( \frac{1}{2} , 1 \right)$. The process $S$ therefore, at time $\vartheta \wedge 1$, takes the value $2$ on a non-negligible part of the set $\{\tau > \frac{1}{2}\}$. Of course, the random variable $S_{\tau}$ equals $2$ on the set $\{\tau\le \frac{1}{2}\}$ too. Hence there is no strong Markovian prescription for the process $S$ what to do after time $\vartheta$: without further information on the past the process $S$ cannot decide whether it should remain constant or continue to move on as a geometric Brownian motion.
\end{proof} 

Let us apply this example to the pricing of options of the form 
$(S_T -x)_+$. Fix $x \ge 2$. It is straightforward to calculate its price $P^{x}(t,z)$ at time $t$, conditionally on $S_t = z$, which is defined via
\begin{equation} \label{cond}
P^{x}(t,z) = \E [(S_T -x)_+|S_t=z],   \qquad \le t\le T, \, z \in \R.
\end{equation}
Letting $z=2$, we find
\begin{equation}\label{case2}
P^{x}(t,z) =0, \qquad 0 <t \le 1.
\end{equation}

On the other hand, for $z \neq 2$ and $\frac{1}{2} \le t <1$, the prices $P^{x}(t,z)$ are given by the usual Black-Scholes formula and therefore strictly positive. Hence, for $\frac{1}{2} \le t <1$, the option prices $z \to P^{x}(t,z)$ are  discontinuous at $z=2$. They also fail to be increasing and convex in the variable $x$ which a reasonable option pricing regime should certainly satisfy. On the other hand we note that these option prices -- strange as they might be -- do not violate the no arbitrage principle as they were legitimately derived from a martingale.

\vspace{1em}

The marginal probabilities of the process $S$ have an atom at the point $2$ which is rather unpleasant. One may ask whether it is possible to construct variants of the above example which have more regular marginals. 

Here is a straightforward modification: fix an uncountable compact $K$ in $\R_+$ with zero Lebesgue measure. For example one may take the classical Cantor set $K = \{1 + \sum_{n=1}^{\infty} \frac{\epsilon_n}{3^n} : \epsilon_n \in \{0,2\}\}$. We can modify the construction of Example \ref{easy} in three steps. For $0 \le t \le \frac{1}{3}$, let $(S_t)_{0 \le t \le \frac{1}{3}}$ be geometric Brownian motion starting at $S_0 =1$. For $\frac{1}{3} \le t \le \frac{2}{3}$, let $(S_t)_{ \frac{1}{3}\le t\le \frac{2}{3}}$ continue to be geometric Brownian motion, but stopped at the first hitting time $\tau$ of the compact set $K$. For $\frac{2}{3} \le t \le 1$ we again distinguish between the cases $\{\tau \le \frac{2}{3}\}$ and $\{\tau > \frac{2}{3}\}$. On the former set the process $S$ remains constant, i.e., $S_t = S_{\frac{2}{3}}$. On its complement $\{\tau > \frac{2}{3}\}$ the process $S$ continues to follow a geometric Brownian motion. The process $S$ thus enjoys all the features of Example \ref{easy} and, in addition, has continuous marginals. Note, however, that these marginals are not given by densities as they are not absolutely continuous with respect to Lebesgue measure.

\vspace{1em}

Turning back to the context of Example \ref{easy} there is another continuous Markovian martingale with the same marginals as $S$, inducing reasonable option prices. In fact, there is a continuous \textit{strong} Markovian martingale with this property and which is unique in this latter class (Theorem \ref{uniqueness theorem} below).

We only give an informal, verbal description of this strong Markov process. For $0 \le t \le \frac{1}{2}$ let $S$ be defined as in Example \ref{easy}. For $\frac{1}{2} \le t \le 1$ we again distinguish two cases. On the set $\{S_t \neq 2\} $ define $S$ to be geometric Brownian motion, but now we stop this motion when $S_t$ hits the value 2. On the other hand, on the set $\{S_t = 2\} $ the process $S$ starts an excursion from $S_t=2$ into the geometric Brownian motion (\ref{geom2}) with a certain intensity rate. We are free to choose this rate in such a way that the mass remaining at the atom $\{S_t=2\}$ equals precisely the constant mass which is prescribed by the given marginals of the process $S$.

We thus have indicared the construction of another continuous martingale having the same marginals as the process $S$ in Example \ref{easy}. One may check that the latter construction is strongly Markovian -- as opposed to the above construction in Example \ref{easy} -- and that the option prices are increasing and strictly convex in the variable $z$ as they should be. It will follow from Theorem \ref{uniqueness theorem} below that the latter martingale is the unique strong Markov solution for the given marginals.


Note that this answers the question raised at the end of the last section: In Lowther's uniqueness theorem, it is not sufficient to consider Markovian (but not necessarily \emph{strongly} Markovian) martingales: as we have just seen, there exist two distinct continuous Markov martingales with the same 1-dimensional distributions. 

 In view of Dupiere's formula, this leads to the next question. It seems natural to conjecture that, provided the call prices are sufficiently regular in $t$ and $x$, there should be only one continuous Markov martingale matching these prices. Correspondingly one would ask:  can one obtain a similar example as above  with absolutely continuous (or even more regular) marginals?

 To the surprise of the present authors it turned out that the answer is ``yes'', even when we pass to the ``most regular'' situation when $S$ is a Brownian motion, i.e., in the Bachelier model (or the Black-Scholes model). The construction is more involved but resting on the above developed intuition, see the accompanying paper \cite{BePaSc21a}. 


\section{Uniqueness of Dupire's diffusion}

There is a huge literature on one-dimensional processes inducing a given family of marginal distributions (see \cite{Ke72, MaYo02, HiRo12, HiRoYo14, BeHuSt16, Lo08b, HaKl07, FaHaKl15, Ho13, Ol08, Al08, BaDoYo11, KaTaTo15} among others). In particular, the late Marc Yor and his co-authors Hirsch, Profeta, and Roynette wrote the beautiful book \cite{HiPrRoYo11} on ``peacocks''.
  This is a  pun on the French acronym PCOC, for ``Processus Croissant pour l'Ordre Convexe'' and a \emph{peacock} is a stochastic process $(X_t)_{t\geq 0}$ for which the family of laws $ \text{law}(X_t), t\geq 0$ is increasing in the convex order. We take here the liberty to use the word peacock also for a family of probabilities $(\mu_t)_{t\geq 0}$ that increases in convex order.\footnote{Fortunately, ``Probabilit\'es Croissant pour l'Ordre Convexe'' still yields the same acronym.}

 To comply with this literature we find it more natural to pass from the multiplicative setting (\ref{eq3.3}) to the additive setting of a martingale diffusion

\begin{equation}\label{eq3.3add}
{dX_t}=\sigma(t,X_t)\, dW_t,\qquad 0\le t \le T.
\end{equation}
Hence we consider now processes taking possibly values in all of $\R$ and switch to the notation $X$ instead of the ``stock price'' $S$. We note, however, that this change is only for notational reasons, and everything below could also be done in the multiplicative setting of the previous sections.

Given a peacock $(\mu_t)_{t\geq 0}$ we may define option prices via
\begin{equation} \label{eq3.1ad}
C(t,x) = \E_{\mu_t}[(X_t-x)_+],
\end{equation}
where $\mu_t, t\geq 0$ denote the marginal probability measures and $x \in \R$. The `multiplicative' formula (\ref{eq3.4}) becomes in the additive setting
\begin{equation}\label{eq3.4add}
 \frac{\sigma^2(t,x)}{2} =  \frac{C_t(t,x)}{ C_{xx}(t,x)} .
\end{equation}

We can now cite Lowther's complete solution of the uniqueness problem within the class of continuous, strong Markov martingales. We stress (and admire) that this theorem does not require any additional regularity assumptions \cite[Theorem 1.2]{Lo08a}.

\begin{theorem}[Lowther] \label{uniqueness theorem}
Let $X= (X_t)_{0\le t\le1}$ and $Y= (Y_t)_{0\le t\le1}$ be  $\R$-valued, {\it continuous, strong Markov martingales}. 
If  $X$ and $Y$ have the same marginal distributions they also have the same joint distributions. 
\end{theorem}

The proof of this theorem is highly technical and its presentation goes far beyond the scope of the present paper. Instead, we formulate a ``toy'' version of the theorem under strong regularity assumptions. We then analyze why the notion of \textit{strong} Markovianity is key in the above theorem and finally give some hints on the strategy of the proof of Theorem \ref{uniqueness theorem}.

\vspace{1em}

Here is the ``toy'' version of this theorem imposing strong regularity assumptions which make life easier:

\begin{assumption}\label{Ass1}
 We suppose that the process $X$ is given by $X_0 = x_0$ and (\ref{eq3.3add}), where $\sigma(t,x)$ is sufficiently smooth to guarantee that there is a unique strong solution $X$. We also suppose that $X_T$ has finite second moment and that the function $C(t,x)$ is strictly convex in the variable $x$, strictly increasing in the variable $t$, and satisfies standard It\^o smoothness assumptions, i.e., twice continuously differentiable in $x$ and once continuously differentiable in $t$. We also assume that, for every $x \in \R$, the pricing function $(t,z) \to P^{X,x}(t,z)$ defined via

\begin{equation} \label{condadd}
P^{X,x}(t,z) = \E [(X_T -x)_+|X_t=z],   \qquad 0 \le t\le T, \, z \in \R.
\end{equation} 
also satisfies these standard It\^o assumptions.

\end{assumption}

These assumptions are strong enough to guarantee that the function $C(t,x)$ indeed satisfies Equation (\ref{eq3.4add}). 



\begin{theorem}\label{baby}
Let $X=(X_t)_{0\le t\le T}$ satisfy Assumption \ref{Ass1}.
Let $Y=(Y_t)_{0\le t\le T}$ be another continuous Markov martingale such that $X_t$ and $Y_t$ have the same distribution, for every $0 \le t \le  T$.
For fixed strike price $x \in \R$, let  $P^{Y,x}(t,z)$ be the corresponding option prices defined via
\begin{equation} \label{condadd'}
P^{Y,x}(t,z) = \E [(Y_T -x)_+|Y_t=z],   \qquad 0 \le t\le T, \, z \in \R,
\end{equation}
and {\bf assume} that, for every $x$, the functions $(t,z) \to P^{Y,x}(t,z)$ also satisfy the above standard It\^o smoothness assumptions.
Then $P^{X,x}(t,z) = P^{Y,x}(t,z)$, for all $t,x,z$ and the processes $X$ and $Y$ have the same joint distributions.
\end{theorem}

\begin{proof}
As the function $(t,z) \to P^{Y,x}(t,z)$ is {\bf assumed} to satisfy the standard It\^o conditions we may apply It\^o's  formula to obtain
\begin{equation}\label{Ito}
dP^{Y,x}(t,Y_t) = P_t^{Y,x}(t,Y_t) \, dt  + P_z^{Y,x}(t,Y_t)\, dY_t + \frac{1}{2} P_{zz}^{Y,x}(t,Y_t)\, d\langle Y {\rangle}_t,
\end{equation}
where $\langle Y {\rangle}_t$ denotes the quadratic variation process of the continuous, square integrable martingale $Y$ (see 
\cite[Theorem 7.6.4]{Ku06}).
By \eqref{condadd'} the process $(P^{Y,x}(t,Y_t))_{0 \le t \le 1}$ is a martingale. The martingale condition implies that the drift term vanishes so that the equality $\frac{1}{2} P_{zz}^{Y,x}(t,Y_t) \, d\langle Y {\rangle}_t = - P_t^{Y,x}(t,Y_t)\, dt$ holds true in the following sense:
\begin{equation}\label{mart}
\frac{1}{2} \int_A  P_{zz}^{Y,x}(t,Y_t)\, d\langle Y {\rangle}_t = - \int_A P_t^{Y,x}(t,Y_t)\, dt,
\end{equation}
whenever we integrate over a predictable set $A \subseteq [0,T] \times \mathcal C[0,T]$. In particular, the finite measure $d\langle Y {\rangle}_t$ on the predictable sigma-algebra of subsets of $[0,T] \times \mathcal C[0,T]$ has to be absolutely continuous with respect to the measure $d\langle X {\rangle}_t = \sigma^2(t,\omega_t)\, dt$. Here we have used the assumption that $t \to C(t,z)$ is strictly increasing in $t$ so  that $\frac{\sigma^2(t,z)}{2} =  \frac{C_t(t,z)}{ C_{zz}(t,z)}$ is strictly positive. We thus may define the function 
\[ \frac{\rho^2 (t,z)}{2} := \frac{P_t^{Y,x}(t,z)}{P_{zz}^{Y,x}(t,z)} \]
so that 
$d\langle Y {\rangle}_t =  \rho^2 (t,Y_t) \, dt$. 
We therefore must have that $Y$ may be represented as in (\ref{eq3.3add}), with $\sigma$ replaced by $\rho$. This implies the relation
\begin{equation}\label{forw}
 C_t(t,x) = \frac{\rho^2(t,x)}{2} { C_{xx}(t,x)} .
 \end{equation}
Comparing with (\ref{eq3.4add}) we obtain $\rho^2 = \sigma^2$ which shows the identity of $X$ and $Y$ in distribution.
\end{proof}

The above theorem provides a sufficient set of regularity assumptions to substantiate the statement in Dupire's paper \cite{Du94}: ``...we can recover, up to technical regularity  assumptions,  a  unique  diffusion  process''.

Of course, one could do some massaging of the above argument to somewhat weaken the very strong Assumptions \ref{Ass1} which we have imposed. But there is a long and thorny road, going far beyond simple cosmetic changes, to arrive at Lowther's  Theorem \ref{uniqueness theorem}.

In Theorem \ref{baby} the strong regularity assumptions implied in particular the strong Markov property of the process $X$. We stress once more that, in the setting of Lowther's theorem \ref{uniqueness theorem}, the \textit{strong} Markov property is the key assumption.
\vspace{1em}

Passing to Lowther's notation and looking at \eqref{condadd}, a crucial step in the above argument is to start from a convex, increasing, and Lipschitz-one function $g(z)$, such as $g(z)=(z-x)_+$, to its conditional expectations 
\begin{equation} \label{conda}
f(t,z) = \E [g(Y_T)|Y_t=z],   \qquad 0 \le t\le T,\, z \in \R.
\end{equation}
In order to start a chain of arguments one has to verify that $f(t,z)$ is a ``nice'' function.
When looking at Example \ref{easy} and its variants we have seen that in this case, for $g(z)=(z-x)_+$, this is not at all the case. 
Its conditional expectation $f(t,z)$ lacked each of the following desired properties: continuity, monotonicity and convexity in $x$. 

Contrary to this lamentable breakdown of regularity, we shall verify in Corollary \ref{corr} that the strong Markov property guarantees that the following three properties  are inherited from $g(\cdot)$ to each $f(t,\cdot)$: convexity, monotonicity, and $1$-Lipschitz continuity (which serves as a more quantitative version of continuity). This preservation of regularity is a decisive feature of Lowther's proof.

%
%
%
%
%
%

\section{Coupling strong Markov processes}

What is the salient property which distinguishes the strong Markov property from the Markov property in our context? While the former condition allows for Lowther's uniqueness theorem to holde true, in Example \ref{easy} we have seen that there may be different continuous Markov martingales inducing the same marginals. The following well-known concept is the key to understanding the difference.

\begin{definition}\label{monotone}
For probability measures $\pi^1$ and $\pi^2$ on $ \R$, we say that $\pi^2$ dominates $\pi^1$ of first order if, for every $a \in \R$, we have $\ \pi^2[a, \infty[\, \ge \pi^1[a, \infty[$.
\end{definition}
 
 We shall show in the next proposition that the strong Markov property of a continuous martingale implies that the transition probabilities $(\pi_x^{s,t})_{x\in \R}$
 \begin{equation}\label{condpro}
 \pi_x^{s,t}[A] = \P [X_t \in A | X_s = x],
 \end{equation}
where $s<t$ and $A$ is a Borel set in $\R$, are increasing of first order in the variable $x$, for every $s<t$.
%
%
%
%
%
%
%
%

We follow Hobson who applied a well-known technique, namely the ``joys of coupling'' (to quote his paper \cite{Ho98c}) in the present context.

\begin{proposition}\label{Hobson}
Let $X = (X_t)_{0 \le t \le T}$ be a continuous strong Markov process with transition probabilities $ \pi_x^{s,t}[ \cdot]$. Then, for $0 \le s \le t \le T$, and for $x < y$ the probability $\pi_y^{s,t}$ dominates  $\pi_x^{s,t}$ of first order.
\end{proposition}

\begin{proof}
Fix $s, t$ and $ x<y $ as above and let $(X^x_u)_{s \le u \le t}$ and $(X^y_u)_{s \le u \le t}$ be independent copies of the process $X$, starting at $X^x_s=x$ and $X^y_s = y$, both defined on the same filtered probability base. Define the stopping time $\tau$ as the first moment $u$ when the process $X^x_u$ equals the process $X^y_u$, if this happens for some $u \in [s,t]$; otherwise we let $\tau = \infty$.

Define the process $\tilde{X}^x$ by
\[
	\tilde{X}^x_u = 
	\begin{cases} 
		X_u^x & \text{for }s \leq u \leq \tau, \\
		X_u^y & \text{for }\tau < u \leq t.
	\end{cases}	
\]
We clearly have 
\begin{equation}\label{ineq}
\tilde{X^x_t} \le X^y_t,\qquad \text{for all } 0 \le t \le T.
\end{equation}
Indeed, if $\tau = \infty$, the paths of $(X^x_u)_{s \le u \le t} = (\tilde{X}_u^x)_{s \le u \le t}$ and $(X^y_u)_{s \le u \le t}$ never touch, so that we even have a strict inequality by continuity of the processes. If $\tau < \infty$, then $\tilde{X}^x$ and ${X^y}$ have ``joined'' at time $\tau$, and follow the same trajectory. Hence $\tilde{X}_u^x = X^y_u$, for $\tau \le u \le t$.

Inequality \ref{ineq} implies that the law of $X^y_t$ dominates the law of $\tilde{X}_t^x$ in first order.
\end{proof}

\begin{corollary}\label{corr}
Let $X = (X_t)_{0 \le t \le T}$ be a continuous strong Markov process with marginal laws $\mu_t$ and transition probabilities $\pi_x^{s,t}[ \cdot]$. Let $0 \le s \le t \le T$  and $z \mapsto g(z)$ be a measurable $\mu_t$-integrable function, and define the conditional expectation similarly as in (\ref{conda})

\begin{equation} \label {condad}
f(s,x) = \E [g(X_t)|X_s=x],   \qquad  x \in \R.
\end{equation}
Then the following assertions hold true.

(i) If $z \mapsto g(z)$ is increasing, then so is $x \mapsto f(s,x)$, for every $0 \le s \le t \le T$.

\noindent If, in addition, we assume that $X$ is a martingale, we also have the following two assertions.

(ii) If $z \mapsto g(z)$ is 1-Lipschitz , then so is $x \mapsto f(s,x)$, for every $0 \le s \le t \le T$.

(iii) If $z \mapsto g(z)$ is convex, then so is $x \mapsto f(s,x)$, for every $0 \le s \le t \le T$.

\end{corollary}

\begin{proof}
Assertion (i): this is just a reformulation of Proposition \ref{Hobson}.

Assertion (ii): If $g$ is 1-Lipschitz, then $g(x) +x$ is increasing. As $X$ is a martingale we have $\E [X_t|X_s=x] = x$. By (i) $f(s,x) + x$ is increasing. By the same token $f(s,x) - x$ is decreasing which readily shows that $x \to f(s,x)$ is Lipschitz.

Assertion (iii): We follow the proof of \cite[Theorem 3.1]{Ho98c}.
For convex $g$ and $x < y < z$ we have to show that
\begin{equation}\label{convex}
(z-x) f(s,y) \le (z-y)f(s,x) + (y-x) f(s,z).
\end{equation}

Fix $x < y < z$ and choose three independent copies $X^x, X^y,X^z$ of the process $X$, starting at time $s$ from the initial values $x,y$, and $z$. To simplify notation we denote the resulting triple of processes $\left(  {X^x},X^y,{X^z} \right)$ by $(X,Y,Z)$. We define coupling times
similarly as above. Let $\tau^x$ be the first moment $u > s$ when $X(u) = Y(u)$; similarly $\tau^z$ is defined as the first moment when $Y$ and $Z$ meet. Finally, let $ \tau = \tau^x \wedge \tau^z \wedge t$. 
This time we leave the processes unchanged; rather we shall argue on the three disjoint (up to null sets) sets $\{\tau = \tau^x\}$, $\{\tau = \tau^z\}$, and $\{\tau = t\}$.

We start with the latter on which we have $X_t < Y_t < Z_t$.
By the convexity of $g$ we have
\begin{equation}\label{convexrandom}
(Z_t-X_t) g(Y_t) \le (Z_t-Y_t) g(X_t) + (Y_t-X_t) g(Z_t).
\end{equation}
so that
\begin{equation}\label{convexexp}
\E\left[(Z_t-X_t) g(Y_t) - \left( (Z_t-Y_t) g(X_t) + (Y_t-X_t)g(Z_t)\right);\tau = t\right] \le 0.
\end{equation}

On $\{\tau = \tau^x\}$ we have $X_t = Y_t $ so that the last term in inequality (\ref{convexrandom}) vanishes. Also the first and the middle term are equal so that (\ref{convexrandom}) holds true (with equality) on the set $\{\tau = \tau^x\}$. In particular,
\begin{equation}\label{convexexp1}
\E\left[(Z_t-X_t) g(Y_t) - \left( (Z_t-Y_t) g(X_t) + (Y_t-X_t)g(Z_t)\right);\tau = \tau^x\right] \le 0.
\end{equation}
The same reasoning applies to $\{\tau = \tau^z\}$. 
\begin{equation}\label{convexexp2}
\E\left[(Z_t-X_t) g(Y_t) - \left( (Z_t-Y_t) g(X_t) + (Y_t-X_t)g(Z_t)\right);\tau = \tau^z\right] \le 0.
\end{equation}
Summing over these three sets we obtain
\begin{equation}\label{convexexp3}
\E\left[(Z_t-X_t) g(Y_t) - \left( (Z_t-Y_t) g(X_t) + (Y_t-X_t)g(Z_t)\right)\right] \le 0.
\end{equation}
Finally we use independence and the martingality of $X,Y$, and $Z$ to obtain
\begin{equation}\label{convex1}
(z-x) \E[g(Y_t)|Y_s=y] \le  (z-y)\E[g(X_T|X_s=x] + (y-x)\E[g(Z_t)|Z_s=z],
\end{equation}
which is tantamount to (\ref{convex}).
\end{proof}


We can reformulate the message of Corollary \ref{corr} (ii) in the spirit of Bachelier by considering the Wasserstein cost $\mathcal{W}_1 (\pi_x^{s,t}[ \cdot] , \pi_y^{s,t}[ \cdot])$ of the horizontal transport of the conditional probability measures $ \pi_x^{s,t}[ \cdot]$ to $ \pi_y^{s,t}[ \cdot]$.

Recall that, for probabilities $\mu, \nu$ on the real line the Wasserstein-1 distance is given by 
\[ \mathcal{W}_1(\mu, \nu):=\inf_{\pi \in \text{cpl}(\mu, \nu)} \int |x-y|\, d\pi(x,y), \]
where $\text{cpl}(\mu, \nu)$ denotes the set of all probabilities on $\R^2$ having $\mu, \nu$ as marginal measures, see e.g.\ \cite{Vi09} for an extensive overview of the field of optimal transport. 

\begin{definition}
Let $\pi$ be a probability on  $\R^2 $ and write $\mu$ for its projection onto the first coordinate and  $(\pi_x)_x$ for the respective disintegration so that $\pi = \int_\R \pi_x d\mu(x)$. Then $\pi$ is called a Lipschitz-kernel if for all $x, y $ in a set $X$ with  $\mu(X)=1$ 
\begin{align}\label{kLip}
\mathcal W_1(\pi_x, \pi_{y})\leq |x-y|.
\end{align} 
\end{definition} 
We call $\pi$ a \emph{martingale coupling} if $\int y\, d\pi_x(y)=x$, $\mu$-a.s. It is then straight-forward to see that for a martingale coupling $\pi$ the following are equivalent: 
\begin{enumerate}
\item $\pi$ is a Lipschitz-kernel.
\item for all $x, y $ in a set $X$ with $\mu(X)=1$ we have
$
\mathcal W_1(\pi_x, \pi_{y})= |x-y|.
$
\item for all $x, y$ in a set $X$ with $\mu(X)=1$, $x\leq y$ the measure $\pi_x$ is dominated in first order by $\pi_{y}$.  
\end{enumerate}

\begin{definition}\label{Wass} 
Let $X$ be an $\R$-valued Markovian process. Then X has the Markov Lipschitz property if for all $s \leq t$, the law of $(X_s, X_t)$ is a Lipschitz kernel. 
\end{definition}
To give yet another characterisation of Lipschitz-Markov processes, recall that a process $X$ is Markov if and only if for all $s\leq t$ and every bounded measurable function $f$ there is a measurable function $g$ such that 
$$\E[f(X_s)|\mathcal F_s]= g(X_t).$$ 
A process $X$ is Lipschitz-Markov if and only if for all $s\leq t$ and every 1-Lipschitz function $f$ there is a 1-Lipschitz function $g$ such that 
$$\E[f(X_s)|\mathcal F_s]= g(X_t).$$
This is a straightforward consequence of the Kantorovich-Rubinstein theorem that provides a dual characterization of the Wasserstein-1 distance through 1-Lipschitz functions. 

We now can resume the crucial role of the strong Markov property

\begin{corollary}\label{SLM}
Let $M$ be a continuous Markov martingale. Then $M$ is Lipschitz-Markov if and only if it is strong Markov. 
\end{corollary}
\begin{proof} 
By Proposition \ref{Hobson} / Corollary \ref{corr} every strong Markov-martingale is Lipschitz Markov.  
That a Lipschitz Markov martingale is strongly Markov is proved in the same way as one establishes the strong Markov property for Feller processes. See e.g.\ \cite[Theorem 1.68]{Li10}. \end{proof}

To the best of our knowledge, Lipschitz kernels play a crucial role in all known proofs of Kellerer's theorem. The decisive property is the following: 
\begin{proposition}\label{LMclosed}
Consider the space of $\mathcal P(\mathcal D[0,1])$ of probability measures on the Skorokhod space equipped with the convergence of finite dimensional distributions.  Then the set of Lipschitz-Markov martingales is closed. 
\end{proposition}
In contrast, the set of Markov-martingales is not closed. 
See e.g.\ \cite{BeHuSt16} for the (simple) proof of Proposition \ref{LMclosed}.

\section{Continuity of the martingale solution}

An important question in the present context is the following: under which conditions on a peacock $(\mu_t)_{0\le t \le 1}$, as defined in Section 4 above, is there a 
strong Markov martingale with \textit{continuous trajectories} having the given marginals. We only focus on the one-dimensional case and mention that the corresponding question for higher dimensions remains wide open. The one-dimensional case, however, is fully understood by now, again by the definitive work of Lowther \cite[Theorem 1.3]{Lo08b}:
\begin{theorem}[Lowther]
Let $(\mu_t)_{t\geq 0}$ be a peacock and assume that $t\mapsto \mu_t$ is weakly continuous and that each $\mu_t$ has convex support. Then there exists a unique strongly continuous Markov martingale $X$ such that such that $X_t\sim \mu_t, t\geq 0$. 
\end{theorem}



We do not show Lowther's theorem in full generality, but again we want to isolate a sufficient set of assumptions that allows us to present a (comparably simple) self-contained proof of the continuity theorem. 

\begin{remark}\label{LMkernelE}
A key ingredient of the proof is that for  probabilities $\mu, \nu$ in convex order there exists a continuous martingale $(X_t)_{0 \le t \le 1}$, $X_0 \sim \mu$, $X_1\sim \nu$  which is strongly Markovian (and hence Lipschitz Markov). For instance we can take $X$ to be a \emph{stretched Brownian motion}, that is, a solution to a continuous time martingale transport problem, see \cite{BaBeHuKa20}. Another possibility would be to apply an appropriate deterministic time change to Root's solution \cite{Ro69} (see \cite{CoWa12} for the case of a non-trivial starting distribution) of the Skorokhod embedding problem. We note that the martingale transport approach is also applicable to measures $\mu, \nu$ defined on $\R^d, d>1$. This could be interesting in view of a possible multi-dimensional  extension of Lowther's theorem but this is not within the scope of the present article.
\end{remark}

\begin{assumption}\label{Ass}

Let $ (\mu_t)_{0 \le t \le 1}$ be a one-dimensional peacock centered at zero, with densities $p_t(x)$ and finite second moments $m^2_2(\mu_t) = \int_{-\infty}^\infty x^2 \, d\mu_t(x) = \int_{-\infty}^\infty x^2 p_t(x)\, dx$ such that the function $ t \to m^2_2(\mu_t)$ is continuous. 

We assume that there is a -- bounded or unbounded -- open interval $I \subset \R$ supporting each $\mu_t$ such that, for each compact subset $ K \subset I$, the Lebesgue densities $p_t(x)$ of $\mu_t$ are bounded away from zero, uniformly in $x \in K$ and $t \in [0,1]$.


\end{assumption}

It will be convenient to suppose (w.l.g.\ via a deterministic time change) that $ t \to m^2_2(\mu_t)$ is affine. More precisely we may assume that $m^2_2(\mu_{t+h}) - m^2_2(\mu_t)= h$ so that for all martingales $M$ with $law( M_t) = \mu_t$ and $law( M_{t+h} )= \mu_{t+h}$

\begin{equation} \label{1/n}
\E \left[\left(M_{t+h} - M_t \right)^2\right] = h.
\end{equation}

\begin{theorem}\label{cont}
Under the above assumptions there is a \underline{continuous} strong Markov martingale $M = (M_t)_{0 \le t \le 1}$ with given marginals $ (\mu_t)_{0 \le t \le 1}$.
\end{theorem}

There is an obvious and well-known strategy for the proof. We want to obtain the desired martingale $M$ as a limit of approximations which fit the peacock $(\mu_t)_{0 \le t \le 1}$ on finitely many points of time. As in \cite{HiRoYo14} it is convenient to do so along the ordered set $\cal S$ of finite partitions $ S = \{s_0,\dots,s_n\} $ of $[0,1]$, where $0 = s_0 < \dots < s_n =1$.

 For each $S \in \cal S$ we choose a continuous strong Markov martingale $M^S$ having the given marginals at each time $s_i \in S$, the existence of  $M^S$  is a direct consequence of Remark \ref{LMkernelE}. 
 
Identifying the martingales $M^S$ with their induced measures on the path space $\mathcal C[0,1]$ this family of measures is tight, if considered on Skorohod space $\mathcal D[0,1]$, equipped with the topology of convergence of finite dimensional distributions.
Hence we can find a cluster point $M$ in the set $\mathcal P(\mathcal D[0,1])$ of probability measures on $\mathcal D[0,1]$, see e.g.\ \cite{BeHuSt16} for the straightforward argument.
By refining the filter $\mathcal S$ we may suppose that $M$ is a limit point. 

We fix such a limiting process $M$ which by  Proposition \ref{LMclosed} is a Lipschitz Markov martingale.
These arguments again are standard by now and, e.g., well presented in the papers \cite{HiRoYo14, BeHuSt16}.
A priori the martingale $M$ has c\`adl\`ag trajectories.
Our present task is to show the \textit{continuity} of the trajectories of the limiting process $M$ under the above assumptions.

We first give a general criterion for the continuity of a limiting martingale $M$ which is somewhat reminiscent of the classical Kolmogorov continuity criterion \cite[Theorem 1.8]{ReYo99}.

\begin{proposition}\label{criterion}
	Let $(M^i)_{i \in I}$ be a net of $\R$-valued continuous strong Markov martingales  $M^i = (M^i)_{0 \le t \le 1}$ and $M$ its limit in the set of probabilities on Skorohod space $\mathcal P(\mathcal D[0,1])$ with respect to convergence in finite dimensional distribution.
	Suppose that there are constants $C_1>0$ and $\beta > 0$ such that, for every $0 \le t_0 < t_0 + h \le 1$ and every $i \in I$
	\begin {equation}\label{BMO}
	\lVert (M^i_t)_{t_0 \le t \le t_0 + h}\rVert_{{BMO}_1} := \sup_\tau\{ \lVert \E[|M^i_{t_0 + h} -M^i_\tau| \mid {\cal{F}}_{\tau}] \rVert_\infty   \}   \le  C_1 h^\beta.
	\end{equation}
	Then the martingale $M$ has continuous trajectories.
\end{proposition}

In (\ref{BMO}) $\tau$ runs through the $[t_0,t_0 + h]$-valued stopping times with respect to the natural filtration of $ (M^i_t)_{t_0 \le t \le t_0 + h}$. 
As $M^i$ is strong Markov, condition (\ref{BMO}) is tantamount to the requirement that the first moments $m_1(\pi_x^{i,\tau ,t_0 +h}) = \int_{-\infty}^\infty |y-x| \, d\pi_x^{i ,\tau ,t_0 +h}(y)$ of the transition probabilities $\pi_x^{i,\tau,t_0+h}$ from $M^i_\tau = x$ to $M^i_{t_0 + h}$ satisfy 
\begin{equation}\label {beta}
 	m_1( \pi_x^{i,\tau,t_0 +h})  = \int_{-\infty}^\infty |y-x| \, d\pi_x^{i ,\tau ,t_0 +h}(y) \le C_1 h^\beta, 
\end{equation}
for $\mu^i_\tau$-almost all $x \in \R$, where $\mu^i_\tau$ denotes the law of $M^i _\tau$.

An important feature of the $BMO$ norms for \textit{continuous} martingales is that, by  the John-Nirenberg inequality, all $BMO_q$ norms are equivalent, for $1 \le q < \infty$ (see, e.g.\ \cite[Corollary 2.1]{Ka06b}). Applying this fact to the present context, inequality (\ref{BMO}) is equivalent to the existence of a constant $C_q > 0$ (for some or, equivalently, for all $1 \le q < \infty$) such that

\begin {equation}\label{BMOq}
\lVert (M^i_t)_{t_0 \le t \le t_0 + h}\rVert_{\text{BMO}_q} := \sup_\tau\left\{ \left\lVert \E[|M^i_{t_0 + h} -M^i_\tau|^q \mid {\cal{F}}_{\tau}]^{\frac{1}{q}} \right\rVert_\infty \right\} \le C_q h^\beta.
\end{equation}

\begin{proof}[Proof of Proposition \ref{criterion}]
Suppose that $M$ fails to be continuous and let us work towards a contradiction to \eqref{BMOq} when $q > \frac{1}{\beta}$.
Assume $M$ has jumps of size bigger than $3a > 0$ with probability bigger than $\kappa >0$, i.e.
\[
	\mathbb P\left( \left\{ \exists t \in [0,1] \colon |M_{t} - M_{t-}| \geq 3a \right\} \right) > \kappa.
\]
As $M$ has c\`adl\`ag paths, there is $h_0 > 0$ such that for all $0 < h \leq h_0$
\[
	\mathbb P\left( \left\{ \exists k \in \mathbb N \colon |M_{(kh)\wedge 1} - M_{((k-1)h)\wedge 1}| \ge 2a \right\} \right) > \kappa.
\]
By the pigeon hole principle, we can find for each $0 < h \leq h_0$ a time $t_0 \in [0,1]$ with
\begin{equation}\label{a}
	\P\left( \left\{ |M_{(t_0 +h)\wedge 1} - M_{t_0}| \ge 2a \right\}  \right) > h\kappa.
\end{equation}
In view of the convergence in finite dimensional distributions of $(M^i)_{i \in I}$ to $M$, we find for each $0 < h \leq h_0$ a time $t_0 \in [0,1]$ (w.l.g.\ $t_0 + h \leq 1$) and an index $i \in I$ with
\begin{equation}\label{a'}
	\P\left( \left\{ |M^i_{t_0 + h} - M^i_{t_0}| \ge a \right\} \right) \geq h \kappa.
\end{equation}
Fixing such an index $i \in I$, it follows that there is a set $A \subset \R$ of positive measure with respect to the law of $M^i_{t_0}$ such that, for all $x \in A$,
\begin{equation}\label{acond}
	\P\left( \left\{ |M^i_{t_0 +h} - M^i_{t_0}| \ge a \mid M^i_{t_0} = x \right\} \right) \geq h\kappa.
\end{equation}
For $x \in A$ we therefore have
\begin{equation}\label{mq}
	m_q(\pi_x^{i ,t_0 ,t_0 +h}) = \left(\int_{-\infty}^\infty |y-x|^q \, d\pi_x^{i,t_0 ,t_0 +h}(y)\right)^{\frac{1}{q}} \geq a(h\kappa )^{\frac{1}{q}}.
\end{equation}
As $q > \frac{1}{\beta}$ we can choose $0 < h \leq h_0$ sufficiently small, with $a (\kappa h)^{\frac{1}{q}} > C_q h^\beta$, and arrive at the desired contradiction to \eqref{BMOq}
\begin{equation*}
	C_q h^\beta \geq \lVert (M_t^i)_{t_0 \leq t \leq t_0 + h} \rVert_{BMO_q} \geq m_q(\pi_x^{i,t_0,t_0+h}) \geq a (h\kappa)^{\frac{1}{q}} > C_qh^\beta.
\end{equation*}
\end{proof}

Turning back to the above defined family $(M^S)_{S \in \cal S}$ of martingales we shall establish an inequality of the type (\ref{beta}) for the transition probabilities $\pi_x^{S ,t_0 ,t_0 +h}$, using the fact that the functions $x \mapsto \pi_x^{S,t_0,t_0+h} $ are Lipschitz-Markov.

\begin{lemma}\label{sublemma}
	Let $(\mu_t)_{0 \le t \le 1}$ be a peacock satisfying Assumption \ref{Ass}.
	Fix a compact set $K \subset I$.
	For all $h > 0$ sufficiently small and $x \in K$, there is a constant $D > 0$ such that for all $S \in \mathcal S$ and $0 \leq t_0 \leq t_0 + h \leq 1$, with $t_0, t_0 + h \in \mathcal S$, the first moments of the transition measures $\pi_x^{S ,t_0 ,t_0 +h}$ can be estimated by
	\begin{equation} \label{m1}
		m_1\biggl(\pi_x^{S ,t_0 ,t_0 +h}\biggr) := \E_\P\left[\left|M^S_{t_0 +h} - M^S_{t_0}\right|  \biggm| M^S_{t_0}=x\right] 
		=\int_{-\infty}^\infty \lvert y-x\rvert \, d\pi_x^{S ,t_0 ,t_0 +h}(y)
		\le D h^{\frac{1}{4}},
	\end{equation}
	where $\P$ denotes the law of the martingale $M^S$.
\end{lemma}

\begin{proof}
We first suppose that $I=\R$. By \eqref{1/n} and Jensen's inequality we have
\[
	\E_\P[|M^S_{t_0+h} - M^S_{t_0}|] \le  h^{\frac{1}{2}}.
\]
 We may rewrite this inequality in the form 
\begin{align} \nonumber
\E_\P[|M^S_{t_0+h} - M^S_{t_0}|] &= \E_\P[\E_\P[|M^S_{t_0+h} - M^S_{t_0}|  \big | M^S_{t_0}]] = \int_{-\infty}^\infty m_1\biggl(\pi_x^{S ,t_0 ,t_0 +h}\biggr) \, d\mu_{t_0}(x) 
\\ \label{n1/2}	 &=    \int_{-\infty}^\infty F(x) \, d\mu_{t_0}(x)\le  h^{\frac{1}{2}},
\end{align}
where, alleviating the notation from $\pi_x^{S ,t_0 ,t_0 +h}$ to $\pi_x$, and
\[ F(x)=\int_{-\infty}^\infty |x - y| \, d\pi_x(y). \]

\textit{Claim:} The function $x \mapsto F(x)$ satisfies the estimate
\begin{equation}\label{oneLip}
	|F(x) - F(x+k)| \le 2k,  \qquad x \in \R,\quad k>0.
\end{equation}
Indeed,

\begin{align*}
	|F(x+k) - F(x)| &= \left| \int_{-\infty}^{\infty} |y - (x + k)| \, d\pi_{x+k}(y) - \int_{-\infty}^{\infty} |y-x| \, \pi_x(y) \right|
	\\ &\leq  k + \left| \int_{-\infty}^\infty |y - (x+k)| \, d\pi_{x}(y) - \int_{-\infty}^\infty |y-x| \, d\pi_x(y) \right| \leq 2k
\end{align*}
proving the claim. 
In the first inequality we have used the fact that $\pi$ is a Lipschitz Markov kernel.

Choose a compact interval $[a,b]$ such that $K \subseteq [a+1,b-1]$ and denote by $l>0$ a lower bound for the density function $p_{t_0}$ of $\mu_{t_0}$ on $[a,b]$.

We want to estimate $F(x_0)$, for $x_0 \in K$, and start by showing the rough estimate $F(x_0) \le  2$.
Using (\ref{oneLip}) we otherwise have $\int_{x_0}^{x_0+1} F(x) \, dx \ge 1$ and we arrive at the following contradiction to (\ref{n1/2}) for small enough $h$
\[ 1 \le \int_{x_0}^{x_0+1} F(x) \, dx  \le \frac{1}{l} \int_{x_0}^{x_0+1}F(x) \, dp_{t_0}(x) dx \le \frac{1}{l} h^{\frac{1}{2}}. \]

If $F(x_0) \le 2$ we may argue similarly, using again
(\ref{oneLip}) and elementary geometry, to obtain, for $x_0 \in K$,
\[ \frac{1}{8} (F(x_0))^2 \le \int_{x_0}^{x_0+1} F(x) \, dx \le \frac{1}{l} \int_{x_0}^{x_0+1}F(x) \, d\mu_{t_0}(x) \le \frac{1}{l} h^{\frac{1}{2}},\]
yielding the desired estimate for $h$ sufficiently small
\[ \sup_{x\in K}   m_1(\pi_x) =   \sup_{x\in K}  F(x) \le D h^{\frac{1}{4}}, \]
where the constant $D>0$ only depends on the compact set $K$, but not on $h$.

Finally we have to come back to our assumption $I=\R$ which allowed us to imbed the compact set $K$ into the interval $[a,b] \subset I$ such that $ K \subset [a+1,b-1]$. If $I$ is only one- or two-sided bounded we have to reason slightly more carefully, as we can imbed the compact set $K$ only into an interval $[a,b] \subset I$ such that $[a + \epsilon,b- \epsilon]$ contains $K$. But no difficulties arise from replacing $1$ by $\epsilon$ and it is straightforward to adapt the above argument also to this situation.
\end{proof}

Under the assumptions of Lemma \ref{sublemma}, Tschebyscheff's inequality and \eqref{m1} allow us to control the difference between medians and means, since for fixed $0<\delta<\frac{1}{4}$
\begin{equation*}
	\pi_x^{S,t_0,t_0 + h}\left( y \colon |y - x| \geq h^{\delta} \right) \leq  Dh^{\frac{1}{4} - \delta}
\end{equation*}
for $h > 0$ sufficiently small, $x \in K$ and feasible $S \in \mathcal S$. 
Hence, we have in the setting of Lemma \ref{sublemma} that
\begin{equation}\label{mean med control}
	\left\lvert\text{median}(\pi_x^{S,t_0,t_0 + h}) - \text{mean}(\pi_x^{S,t_0,t_0+h})\right\lvert \leq h^{\delta}.
\end{equation}

\begin{lemma} \label{stop}
Let $0 < \delta < \frac{1}{4}$.
Under the assumptions of the Lemma \ref{sublemma} the same conclusion as in (\ref{m1}) holds true for every $[t_0,t_0+h]$-valued stopping time $\tau$ (by possibly changing the constant $D$ to a different $C$) for $x \in K$
\begin{equation} \label{mtau}
	m_1\biggl(\pi_x^{S ,\tau ,t_0 +h}\biggr) := \E_\P\left[\left|M^S_{t_0 +h} - M^S_{\tau}\right|  \biggm| M^S_{\tau}=x\right] 
	=\int_{-\infty}^\infty \lvert y-x\rvert \, d\pi_x^{S ,\tau ,t_0 +h}(y)
	\le C h^{\delta}.
\end{equation}
\end{lemma}

\begin{proof}
	By Corollary \ref{SLM} we have that $x \mapsto \pi^{S,t_0,t}$ is Lipschitz Markov for all $s \in [t_0,t_0+h]$.
	From this, we can deduce the continuity of the map
	\[
		(x,t) \mapsto \int_{-\infty}^{\infty} |z-y| \, d\pi^{S,t,t_0 + h}(z).
	\]
	Therefore, it suffices to show \eqref{mtau} for deterministic $\tau \equiv t \in [t_0,t_0 + h]$ due to the strong Markov property.
	To this end, let $\tilde K$ be a compact interval in $I$, containing the compact set $K$ in its interior and fix the constant $D$ from Lemma \ref{sublemma}, applied to $\tilde K$.

	To do so, find $x \in I$ such that $y$ equals the median of the measure $\pi_x^{S,t_0,t}$, that is
	\[ \pi^{S,t_0,t}_x(]-\infty, y[) \leq \frac{1}{2} \leq \pi^{S,t_0,t}_x(]-\infty, y]).  \]
	Since $x \mapsto \pi_x^{S,t_0,t}$ is Lipschitz Markov by Corollary \ref{SLM}, such an $x$ exists.
	We may use $x$ to obtain the following estimate
	\begin{multline}\label{+}
		\E_\P\left[\left|M^S_{t_0 +h} - y \right| \biggm| M^S_{\tau}=y\right] 
		\\ \le \E_\P\left[\left(M^S_{t_0 +h} - y\right)_+  \biggm| M^S_{\tau} \ge y, M^S_{t_0} = x\right] + \E_\P\left[\left(M^S_{t_0 +h} - y \right)_-  \biggm| M^S_{\tau} \le y, M^S_{t_0} = x\right]
		\\ \le 2 \E_\P\left[ \left| M^{S,t_0,t_0+h} - y \right| \biggm| M^{S}_{t_0} = x \right] \\ \le 2 \E_\P\left[ \left| M^{S,t_0,t_0+h} - x \right| \biggm| M^{S}_{t_0} = x \right] + 2|x - y|.
	\end{multline}
	Here, we used the first item of Corollary \ref{corr} for the first inequality and $y$ being the median of $\pi^{S,t_0,\tau}_x$ for the second.
	Note that for $h$ sufficiently small we have by \eqref{mean med control} that $x \in \tilde K$. 
	Applying the estimates \eqref{m1} and \eqref{mean med control} to \eqref{+} we find
\[
	\E_\P\left[\left|M^S_{t_0 +h} - y \right| \biggm| M^S_{\tau}=y\right] \le 2D h^\frac{1}{4} + 2h^\delta \leq \left( 2D + 2 \right) h^\delta,
\]
which concludes the proof.
\end{proof}

\begin{proof}[Proof of Theorem \ref{cont}]
	We have to combine Proposition \ref{criterion} and Lemma \ref{stop} with a stopping argument.
	To this end, let $(K_n := [a_n,b_n])_{n \in \mathbb N}$ be an increasing sequence of compact intervals exhausting the interval $I = ]a,b[$, where $a,b \in [-\infty, \infty]$. We let $K_1 = \emptyset$.
	Define the stopping times $\tau^{S,n}$ as the first moment when the continuous martingale $M^S$ leaves $K_n$, write $M^{S,n}$ for the stopped process, and denote the differences $M^{S,n+1} - M^{S,n}$ by $\Delta M^{S,n}$.

	For the processes $\Delta M^{S,n}$ the assumptions of Proposition \ref{criterion} still hold true as a consequence of Lemma \ref{stop} and the strong Markov property of $M^S$.
	Consider the process
	

	\[
		\tilde M^S_t := (\Delta M_t^{S,n})_{n \in \mathbb N}, \quad t \in [0,1],
	\]
	taking values in $\R^\mathbb N$.
	Moreover, let 
	\[
		m^S := \inf \left\{ k \in \mathbb N \colon \left\lvert \Delta M^{S,k} \right\rvert_\infty = 0 \right\}	
	\]
	be the smallest integer such that the entire trajectory of $(M^S_t(\omega))_{0 \le t \le 1}$ is contained in $K_{m^S}$.
	We know already that for every $n\in\mathbb N$, $M^{S,n}$ (and therefore $\Delta M^{S,n}$) allows for f.d.d.\ convergent subnets along $\mathcal S$ where the limits have by Proposition \ref{criterion} continuous versions.
	We want to argue that similarly the pair $(\tilde M^S,m^S)_{S\in\mathcal S}$ taking values in $\mathcal C[0,1]^\mathbb N \times \mathbb N$ admits a convergent subnet, too.
	For this reason, it is sufficient to show the following.

	\emph{Claim:} For every $\epsilon > 0$ there is $n \in \mathbb N$ such that, for every $ S \in \mathcal S$,
	\begin{equation} \label{cont claim}
		\sup_{S \in \mathcal S} \mathbb P\left( m^S > n \right) < \epsilon.
	\end{equation}
	 Indeed, for each $n\in \mathbb N$ (sufficiently large), there are maximal positive numbers $\alpha_n,\beta_n$ with $\alpha_n + \beta_n \le 1$ such that the probability measure
	\[
		\alpha_n \delta_{a_n} + \beta_n \delta_{b_n} + (1 - \alpha_n - \beta_n) \delta_{\frac{\text{mean}(\mu_1) - \alpha_n a_n - \beta_nb_n}{1 - \alpha_n - \beta_n}}
	\]
	is dominated by $\mu_1$ in convex order.
	Since $\mu_1$ puts no mass to the boundary of $I$, there is for each $\epsilon > 0$ an index $N \in \mathbb N$ such that $\alpha_n + \beta_n < \epsilon$ for all $n \ge N$.
	The law of $M^{S,n}$ is dominated in the convex order by $\mu_1$.
	By maximality of $\alpha_n$ and $\beta_n$ we find uniformly for all $S \in \mathcal S$
	\[
		\mathbb P\left( m^{S} > n \right) = \mathbb P\left( \tau^{S,n} < \infty \right) = \mathbb P\left( M^{S,n} \in \{a_n,b_n\} \right) \leq \alpha_n + \beta_n < \epsilon,
	\]
	which yields the claim \eqref{cont claim}.

	By passing to a subnet, still denoted by $\mathcal S$, we thus obtain that $(\tilde M^S, m^S)_{S \in \mathcal S}$ admits a limit $(\tilde M,m)$ w.r.t.\ convergence of finite dimensional distributions, where
	\[ \tilde M_t = (\Delta M_t^n)_{n \in \mathbb N},\qquad t \in [0,1], \]
	and $(\Delta M_t^{S,n})_{S \in \mathcal S}$ has $\Delta M_t^n$ as its f.d.d.\ limit.
	Due to Proposition \ref{criterion} we may choose a version of $\tilde M$ taking values in $\mathcal C[0,1]^\mathbb N$.
	Consider the following process
	\[
		\hat M_t := \sum_{n = 1}^m \Delta M_t^n,	\qquad t \in [0,1],
	\]
	which has continuous trajectories as $m$ is integer-valued finite.
	Note that f.d.d.\ convergence of $\tilde M^S$ to $\tilde M$ yields f.d.d.\ convergence of
	\[
		M^S = \sum_{n = 1}^{m^S} \Delta M^{S,n} \to \sum_{n = 1}^m \Delta M^n = \hat M.
	\]
	We conclude that $\hat M$ and $M$ coincide in law, and thus the Lipschitz Markov martingale $M$ has a version with continuous paths.
\end{proof}

\section{Overview of related results in the literature} In the last sections we have focused on specific aspects of the theory of mimicking processes. 
In this final section we provide an overview  of related results in the literature and give some context for theorems discussed above. 

An early influential result is the work of Strassen \cite{St65} who established that there exists a submartingale with marginals $\mu_1, \ldots, \mu_N$ if and only these measures are increasing in the increasing convex order induced by increasing convex functions. He also proved that there exists a martingale with values in $\R^d$ and marginals $\mu_1, \ldots, \mu_N$  if and only if these measures are increasing in convex order. 
In \cite{Ke72, Ke73}, Kellerer managed to extend Strassen's result on the existence of submartingales to an arbitrary family of marginals. As an important particular case, this yields the existence of a mimicking martingale if the marginals increase in convex order.
Over time a number of authors have given new approaches to Kellerer's theorem see \cite{HiRo12, HiRoYo14, Lo08a, Lo08b, Lo08d, BeHuSt16, BeJu21}. As discussed extensively above the work Lowther adds substantial new developments. He  characterizes when the Markov martingale can be chosen to be continuous, as well as adding a clear-cut uniqueness part to  Kellerer's original result, complementing the formal uniqueness result of Dupire \cite{Du94}. 
We also recall from  above, that the question whether the natural extension of Kellerer's result to higher dimension holds true remains completely open. 


\medskip

Given a continuum of marginals which increase in convex order (and maybe satisfies additional technical conditions), different authors have provided specific constructions of (not necessarily Markovian) martingales that match these marginals. A main motivation stems from the calibration problem in mathematical finance. An additional goal has often been to give constructions that optimize particular functionals given the martingale and marginal  constraints since this yields robust bounds on option prices. Madan and Yor \cite{MaYo02} and 
K\''allblad, Tan, and Touzi \cite{KaTaTo15} establish a time continuous version of the Azema-Yor embedding. 
Hobson \cite{Ho16} established a continuous time version of the martingale  coupling constructed in \cite{HoKl13}.
Henry-Labordere, Tan, and Touzi \cite{HeTaTo16} as well as
Br\''uckerhoff, Huesmann, and Juillet \cite{BrHuJu20} provide continuous time versions of the shadow coupling (originally introduced in \cite{BeJu16}). Richard, Tan, and Touzi \cite{RiTaTo20} give a continuous time version of the Root solution to the Skorokhod embedding problem.
In a slightly different but related direction, Boubel and Juillet  \cite{BoJu18} consider a continuum  of marginals on the real line that do not satisfy an order condition and construct a canonical Markov-process matching these marginals.  We also refer to the book \cite{HiPrRoYo11} of Hirsch, Profeta, Roynette, and Yor that collects a variety of related constructions.

The problem of finding martingales with given marginals has received specific attention in the case where these marginals equal the ones of Brownian motion. Hamza-Klebaner \cite{HaKl07}  posed the challenge of constructing martingales with Brownian marginals that differ from Brownian motion, so called \emph{fake Brownian motions}. 
 Non-continuous solutions can be found in  Madan-Yor \cite{MaYo02}, Hamza-Klebaner \cite{HaKl07}, Hobson \cite{Ho13}, and Fan-Hamza-Klebaner \cite{FaHaKl15} whereas continuous (but non-Markovian) fake Brownian motions were constructed by Oleszkiewicz \cite{Ol08}, Albin \cite{Al08}, Baker-Donati-Yor \cite{BaDoYo11} and Hobson \cite{Ho16}.
As already noted,  the accompanying article \cite{BePaSc21b} establishes that there exists  a Markovian martingale with continuous paths that  has Brownian marginals. 

\medskip

A somewhat different direction arises if one starts with marginals that do not merely satisfy a structural condition (specifically,  monotonicity in  convex order) but rather assumes that a set of marginals is generated from an Ito-diffusion
\begin{align}\label{SourceSDE} dX_t = \sigma_t \, dB_t + \mu_t\, dt\end{align}
and one seeks a Markovian diffusion 
$$ d\hat X_t = \hat \sigma_t (X_t) \, dB_t + \mu_t(X_t)\, dt$$
that mimicks the evolution of $X$ in the sense that $X_t \sim \hat X_t$ for each $t\geq 0$. The process $\hat X$ is then considered to be the \emph{Markovian projection} of $X$. This line of research goes back essentially to the work of Krylov \cite{Kr85} and Gy\''ongy \cite{Gy86}. Of course, also the work of Dupire \cite{Du94} can be seen as a formal contribution to this line of research.
A rigorous justification of Dupire's formula under rather general assumption is obtained by Klebaner \cite{Kl02}. 
A very general theorem on mimicking aspects of Ito processes is given by Brunick and Shreve \cite{BrSh13}. Recently, Lacker, Shkolnikov, and Zhang \cite{LaShZh20} show that the results of \cite{Gy86, BrSh13} can be established directly from the superposition principle of Trevisan \cite{Tr16} (or Figalli \cite{Fi08} in the case where \eqref{SourceSDE} has bounded coefficients). (Notably, the main focus of the work  \cite{LaShZh20} is a mimicking result that shows that conditional time marginals of an Ito-process can be matched by a solution of a conditional McKean-Vlasov SDE with Markovian coefficients.)


In the mathematical finance community, Markovian (local vol) models are often considered to exhibit dynamics that are not particularly realistic. There has been significant interest to combine the convenience that the local vol model offers in terms of calibration with more realistic dynamics that are exhibited by other classes of financial models. That is,  given a Markovian model $d\hat X_t = \hat \sigma_t (\hat X_t)\, dS_t$ that represents marked data, one would like to ``reconstruct'' a more  realistic model $d X_t = \sigma_t \, dB_t$ and thus to ``invert'' the Markovian projection. A concrete way to perform this inversion is the \emph{stochastic local volatility model} see the work of Guyon and Henry-Labordere \cite{GuHe11, GuHe12, GuHe13}. However, it is remarkably delicate to establish existence and uniqueness results for the resulting SDEs.
Partial solutions where given by 
Jourdain and Zhou \cite{JoZh16} and by Lacker Shkolnikov and Zhang in  \cite{LaShZh19}. The problem is also discussed by Acciaio and Guyon \cite{AcGu20} who consider it an important open problem to establish existence of the stochastic local volatility model under fairly general assumptions.


\bibliographystyle{abbrv}
\bibliography{../MBjointbib/joint_biblio}

\end{document}